\documentclass[a4paper, 10 pt]{article}  % Comment this line out if you need a4paper

\usepackage{epsfig} % for postscript graphics files
\usepackage{graphicx}
\usepackage{amsmath}
\usepackage{amssymb}
\usepackage{amsfonts}
\usepackage{bbm}
\usepackage{amsthm}
\theoremstyle{plain}
\newtheorem{thm}{Theorem}
\newtheorem{prop}[thm]{Proposition}

\theoremstyle{definition}

\newtheorem{assum}{Assumption}
\theoremstyle{remark}

\newcommand{\R}{\mathbbm{R}}

\newcommand{\so}[1]{\mathfrak{so}\left( #1 \right)}
\newcommand{\SO}[1]{SO\left( #1 \right)}

\newcommand{\trace}{\text{Tr}}
\def\deff{\stackrel{\triangle}{=}}

\newcommand{\ncom}{\newcommand}
\ncom{\beqn}{\begin{eqnarray*}} \ncom{\eeqn}{\end{eqnarray*}}
\ncom{\beq}{\begin{eqnarray}} \ncom{\eeq}{\end{eqnarray}}
\newcommand\covdg{\stackrel{\mathbb{G}}{\nabla}}

\newcommand\covdgi{\stackrel{g}{\nabla}}
\newcommand\Ad{\operatorname{Ad}}
\title{\LARGE \bf Position and line-of-sight stabilization of spherical robot  using feedforward proportional-derivative geometric controller}
\author{Krishna Chaitanya Kosaraju, Arun D.
Mahindrakar, Vijay Muralidharan\\Anup K. Ekbote and Ramkrishna Pasumarthy
\thanks{Chaitanya, Vijay and Anup   are  graduate students in the Department of Electrical  Engineering,
       Indian Institute of Technology Madras, Chennai-600036,
      India. {\tt\small  kkrishnachaitanya89@gmail.com, m\_vijay\_india@yahoo.co.in, anupekbote@gmail.com}}
\thanks{Arun and Ramkrishna are   with the Department of Electrical Engineering,  Indian Institute of Technology Madras, Chennai-600036,
India {\tt\small arun\_dm,ramkrishna@iitm.ac.in}}}

\begin{document}
\maketitle
%\thispagestyle{empty} \pagestyle{empty}
%
%\tableofcontents \thispagestyle{empty}
\begin{abstract}
In this paper we present a geometric control law for position and
line-of-sight   stabilization of the nonholonomic spherical robot
actuated by three independent actuators. A simple configuration
error function with an appropriately defined transport map is
proposed to extract feedforward and proportional-derivative control
law. Simulations are provided to validate the controller
performance.

\end{abstract}
%
% \cite{Arnold}
\section{Introduction}
The application of Lie groups in Mechanics has been the subject of interest
to the control community as it provides a rich platform for the
application of geometric control techniques. The textbook
\cite{bullo}, provides comprehensive treatment of geometric methods
for mechanical systems defined on manifolds. In \cite{Jason}, the
authors present a geometric PD controller for a double-gimbal
mechanism that evolves on the torus. An output tracking for
aggressive maneuvers involving various flight modes is presented in
\cite{uav} for an unmanned quadrotor.  Mechanical systems when subjected to motion constraints, particularly nonholonomic was presented in \cite{bloch}.  In this paper, we consider a
nonholonomic mechanical system  involving the spherical robot
rolling on a horizontal plane.

The control design for spherical robot initiated with motion
planning and open-loop steering input designs with Euler-angle
parameterizations. A few notable examples are
\cite{sph23,sph13,sph1}. The study of the geometric properties of
spherical robot is a recent interest. A steering control for full
state reconfiguration based on the geometry of the sphere was
proposed in \cite{sph40}. Euler-Poincar\'{e} equations using a
coordinate-free approach were obtained in \cite{sph3,sph4,sph9} for
various actuator configurations. Geometric open-loop control
algorithms were developed in \cite{sph3} for steering the spherical
robot to the origin. Stabilizing control inputs were designed in
\cite{sph4} using the geometric model of the spherical robot for two
independent objectives, a finite-time position stabilization and a
finite-time attitude stabilization.

The control laws reported in literature are obtained by observations
on the mathematical model of the spherical robot, we intend to
identify a control objective which can be accomplished by the
currently established tools in geometric control design
\cite{bullo}. The negative result of Brockett \cite{brockett}  for
nonholonomic systems rules out asymptotic stabilization to an
equilibrium point using smooth geometric control laws. We identify
that position and line-of-sight stabilization problem is achievable
within the framework of smooth geometric control. The notion of
configuration error function and the associated transport map are
the necessary prerequisites  in applying the geometric tools
developed in \cite{bullo}. In this direction, we propose a novel
potential function for the spherical robot model to meet the control
objective of position and line-of-sight stabilization. In doing so,
we design a transport map that paves the way for the synthesis of a
feedforward proportional-derivative geometric control law.

\section{Preliminaries}
%----------------------------------
%
Let the orientation of a rigid body be denoted by $R(t) \in \SO{3}$
relative to the reference inertial frame, where $\SO{3} =\{R|R^\top R=I, det(R)=1 \}$. $\dot{R}(t) \in T_R\SO{3}$, the tangent space to  $\SO{3}$ at $R$. $\SO{3}$ is a Lie group and
$T_{I}\SO{3}\simeq \so{3}$ is the Lie algebra of the group, where
$I$ is the identity element of the group $\SO{3}$, $\so{3}$ is a
vector space formed by skew-symmetric matrices. Since $so(3)$ is
isomorphic to $\mathbb{R}^3$, we denote wedge operation by
% the wedge operator '$\wedge$', $x \in
%\R^3 \mapsto \hat{x} \in \so{3}$ defined as

 \beqn
 \hat{x} =
 \left [\begin{array}{ccc}
0 & -x_3 & x_2 \\
x_3 & 0 & -x_1 \\
-x_2 & x_1 & 0
\end{array} \right]
\label{wedge} \eeqn
for $x \in \R^3$.
% with the bracket structure
%
%\beqn [\hat{\omega},\;\hat{v}] =
%\hat{\omega}\hat{v}-\hat{v}\hat{\omega}, \;\;\;\;  \forall\;
%\hat{v},\; \hat{\omega} \in \so{3}. \eeqn
%
 Further, $\vee $ be the
inverse of the wedge operation and the Lie algebra isomorphism
between $(\R^3, \times)$ and $(\so{3},[\cdot,\;\cdot])$ is
 \beq [\hat{\omega},\;\hat{v}]^{\vee}=\omega \times v, , \;\;\;\;  \forall\;
v,\; \omega \in \R^3.\label{so3_R3_iso} \eeq
 The dual of $\so{3}$ can be identified with $\mathbb{R}^3$ using the
map $\wedge^{\ast}:\so{3}^{\ast}\rightarrow \mathbb{R}^3$. For
$\eta\in \so{3}^{\ast}$ and $\hat{\rho}\in \so{3}$, the action of
$\eta$ on $\hat{\rho}$ can be identified with the usual inner
product $\textquoteleft\cdot\textquoteleft$ in $\mathbb{R}^3$ as
$\eta(\hat{\rho})=\wedge^{\ast}(\eta)\cdot \rho$.
 Let $R,R_1 \in
\SO{3}$, the left translation map  $L_R: \SO{3}\rightarrow \SO{3}$
is defined as $L_R(R_1)=RR_1.$
In a similar way, the right translation map $R_R: \SO{3}\rightarrow
\SO{3}$ as
 $R_R(R_1)=R_1R.$ %for all $g, h \in G$.``
%\subsection{Left and Right Invariant vector fields}
From here unless stated as constant, all the variable are assumed to be time varying. A vector field $X(R) \in T_R\SO{3}$ is left invariant if $
X(RR_1)=RX(R_1), $ and similarly right invariant if
$X(R_1R)=X(R_1)R. $

Body angular velocities of a rigid body are left invariant vector
fields, while the spatial angular velocities are right invariant.
They can be identified using their velocity at the group identity
$I$ of $\SO{3}$. Let $R\in \SO{3}$, $X(R)\in T_RSO(3)$, $X(I)=\hat{v}\in T_I\SO{3}\simeq
\so{3}$.
%and $\dot{R}\in T_R\SO{3}$.
 If $v$ is body angular velocity
then $X(R)=RX(I)=R\hat{v}$, while if $v$ is spatial angular velocity
then $X(R)=X(I)R=\hat{v}R$.
 The velocity  $\dot{R}=R\hat{v}$ at  point $R$, which is equivalent to $T_R\dot{R}$,  can be defined
 using the map $T_IL_R:\so{3}\rightarrow T_R\SO{3}$ as
$T_IL_R\hat{v}$.  Accordingly, the dual of $T_IL_R$ is the map $(T_I
L_R)^{\ast}:T_R\SO{3}^{\ast}\rightarrow \so{3}^{\ast}$. Let
$\beta_R\in (T_R\SO{3})^{\ast}$. Then the action of $\beta_R$ on
$T_IL_R\hat{\omega}$ can identified with the inner product
$<\cdot,\cdot>_{\trace}$ by
$<(T_IL_R)^{\ast}\beta_R,\hat{\omega}>_{\trace}$, where $<\cdot ,
\cdot>_{\trace}$ on $\mathbb{R}^{n \times n}$ is defined as
$<A,B>_{\trace} =\frac{1}{2}\trace(A^{\top}B)$ for $A,B \in
\mathbb{R}^{n \times n}$.

The Riemannian metric $\mathbb{G}(R): T_R\SO{3}\times T_R\SO{3}
\rightarrow \mathbb{R}$, a $(0,2)-$tensor on $\SO{3}$ defined as
$\mathbb{G}(R)(X(R),Y(R))=X(R)^{\top}\mathbb{G}(R)Y(R)$
is left invariant if
 \beqn
\mathbb{G}(R)(X(R),Y(R))=\left(R\mathbb{G}(I)R^{-1}\right)(X(R),Y(R))\eeqn
where $X(R),\;Y(R)\in T_RSO(3)$. Therefore it can be seen that for left invariant vector fields
$X(R), Y(R)$,
 \beq
&&\mathbb{G}(R)(X(R),Y(R))\nonumber\\
&&=\left(R\mathbb{G}(I)R^{\top}\right)(X(R),Y(R))\nonumber  \\
&&=\left(R\mathbb{G}(I)R^{\top}\right)(RX(I),RY(I))\nonumber  \\
&&=\left(R^{\top}\left(R\mathbb{G}(I)R^{\top}\right)R\right)(X(I),Y(I))\nonumber \\
&&=\mathbb{G}(I)(X(I),Y(I))
\label{left invariant Rmetric} \eeq
which is a constant. Since $X(I), Y(I) \in T_I\SO{3}\simeq \so{3}$,
$J
\deff \mathbb{G}(I)$, a $(0,2)-$tensor on $\so{3}$.

For $\hat{w}\in \so{3}$ the adjoint map  $%
\Ad: \SO{3} \times \so{3}\rightarrow \so{3}$ is defined as
\beq
\Ad_{R}(\hat{w})=R\hat{w}R^{\top} =\widehat{(Rw)}.
 \eeq
The following general facts involving matrix operations will be
useful. For $A, B, C \in \mathbb{R}^{n \times n}$, we denote the
trace of $A$ as $\trace(A)$, the symmetric component of $A$ by
$sym(A)=\frac{A+A^{\top}}{2}$ and the skew-symmetric component as
$skew(A)=\frac{A-A^{\top}}{2}$ and if $A=A^{\top}$, $B=-B^{\top}$
then $\trace(AB)=0$. For $a,b \in \mathbb{R}^3$,
$\trace(\hat{a}\hat{b})=-2(a^{\top}b)$. It then follows that
\beqn
\begin{array}{lcl}
\trace(C \hat{a})&=& \trace((sym(C)+skew(C))\hat{a})  \nonumber\\
&=& \trace((sym(C))\hat{a})+ \trace(skew(C))\hat{a}) \\
&=& 0 + \trace(skew(C))\hat{a})  \\
&=& -2 ((skew(C))^{\vee} \cdot\; a)
\end{array}
\eeqn
Therefore $<\hat{a},\hat{b}>_{\trace} = a\cdot b$.
\section{Modeling of spherical robot}
\begin{figure}[h!]
\begin{center}
\includegraphics[width=.6\linewidth]{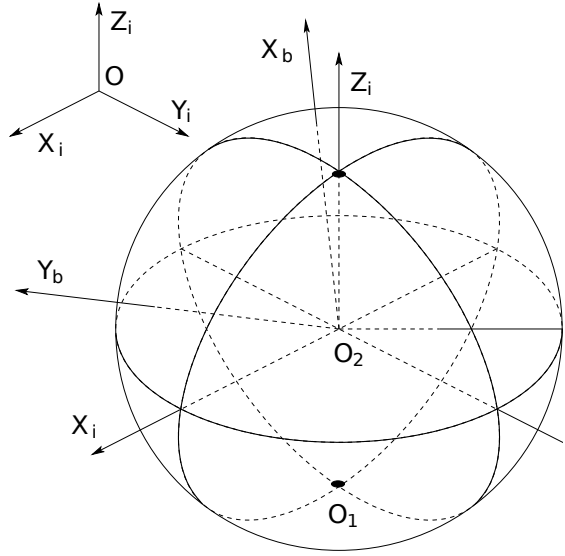}
 \caption{Schematic of the spherical robot} \label{schematic}
\end{center}
\end{figure}

The spherical robot schematic shown in Figure \ref{schematic}
consists of a spherical shell of radius $r$ and mass $m$ moving in a
horizontal plane. The center-of-mass of the robot is assumed to
coincide with the geometric center. The position coordinates of the
spherical robot are denoted by $(x,y)$, which are the coordinates of
the point $O_1$ with respect to $O$.
  Let
$J=\text{diag}(J_1,J_2,J_3) \in\mathbb{R}^{3\times 3}$ be the
 moment-of-inertia matrix of the robot with respect
to the body frame centered at $O_2$. We make the following
assumption.
\begin{assum}
The principal moments of inertia satisfy $0<J_1<J_2<J_3$.
\label{assum}
\end{assum}
 The sphere has three independent
torques acting on the body-coordinate frame. The orientation of body
frame $(X_b,Y_b,Z_b)$ of the robot with respect to an inertial frame
$(X_i,Y_i,Z_i)$ is given by a matrix $R\in \mathrm{SO}(3)$.
The no-slip constraints are given by
\beq v=\left[ \begin{array}{c} \dot{x} \\ \dot{y} \\ \dot{z}
\end{array} \right] = \Omega \times \left[ \begin{array}{c} 0 \\ 0
\\ r \end{array} \right] =r \;R\omega\times e_3, \label{cons1} \eeq
where, $\omega\in \mathbb{R}^3$ denotes the body angular velocity
and $\Omega\in \mathbb{R}^3$ is the spatial angular velocity of the
robot. Denoting the rows of $R$ by $r_1, r_2, r_3$, the kinematics
of the spherical robot is given by
\beq
\begin{array}{lcl}
\dot{x}&=&r(\omega \cdot r_2)\\
\dot{y}&=&-r(\omega \cdot r_1)\\
\dot{R}&=&R\widehat{\omega}.
\end{array}
\label{kin} \eeq
Let $X,Y \in T_R\SO{3}$, an Levi-Civita affine connection on $\SO{3}$ is left
invariant if it satisfies
\beq
\covdg_{T_IL_RX(I)}T_IL_RY(I)&=&T_IL_R\covdg_{X(I)}Y(I)
\label{leftic}\eeq
for all $R \in \SO{3}$  and let $\{e_1, e_2,
e_3\}$ span $\R^3$. Since $\R^3$ is naturally isomorphic to
$\so{3}$, it implies that $\mathrm{span}\{\hat{e}_1, \hat{e}_2,
\hat{e}_3\}= \so{3}$. It then follows for $X^i, Y^i \in \R$
we define $X(R)=X^i\hat{e}_i(R), Y(R)=Y^i\hat{e}_i(R)$ and
\eqref{leftic} can be simplified as follows
\beqn
\begin{array}{lcl}
&&\covdg_{X(R)}Y(R)\\
&&=T_IL_R\covdg_{X(I)}Y(I)\\
&&=R.\left(\covdg_{X^i\hat{e}_i(I)}Y^i\hat{e}_j(I)\right)\\
&&=R.\left(   DY(I).X(I) +
X^iY^j\covdg_{\hat{e}_i(I)}\hat{e}_j(I)\right) \end{array}
\label{lic1}
 \eeqn
where $DY$ is Jacobian of $Y$. From \eqref{left invariant Rmetric} we see that $\mathbb{G}(I)\simeq
J$ represents an inner product on $\so{3}$, and
$G(I)(\hat{e}_i,\hat{e}_j)$ has a constant value which renders
$\covdgi :so(3) \times so(3)\mapsto so(3)$  a bilinear map. It now follows as
\beq &&\covdg_{X(R)}Y(R)\nonumber\\
&&=T_IL_R\left(DY(I).X(I)+X^iY^j\covdgi_{\hat{e}_i}\hat{e}_j\right)\nonumber \\
&&=T_IL_R\left(DY(I).X(I)+\covdgi_{X^i\hat{e}_i}Y^j\hat{e}_j\right)\nonumber \\
&&=T_IL_R\left(DY(I).X(I)+\covdgi_{X(I)}Y(I)\right) \label{lic}
\eeq
 In \eqref{lic}, we
observe that $X(I)$, $Y(I)$ and $\covdgi_{X(I)}Y(I) \in \so{3}$. By
letting $\hat{\omega}=\omega^i\hat{e}_i$,
$\dot{R}=T_IL_{R}\hat{\omega}(t)=\omega^i\hat{e}_i(R)$,
where $\omega^i\in \R$ also known as pseudo velocities. Let
$\hat{\tau}\in (\so{3})^{\ast}$ be  the covector representing the
external torque acting on the robot.
Next,  the covariant derivative of $\dot{R}$ is
 \beq
 \begin{array}{lcl}
\covdg_{\dot{R}}\dot{R}&=&R\left(\frac{d}{dt}\hat{\omega}+\covdgi_{\hat{\omega}}\hat{\omega}\right)\\
&=&T_IL_R J^{-1}\hat{\tau}. \end{array}\label{cov_der_R}
\eeq
 From \eqref{cov_der_R}, we obtain the well-known attitude
dynamics governed by Euler-Poincar\'{e} equations of motion
 \beq \dot{\omega} &=& -J^{-1} (\omega\times J \omega ) + J^{-1}\tau\label{dyn} \eeq
where $\tau\in \R^3$ is the external torque about the body-axis of
the robot.
%
%----------------------------------------------------------------------------------
\section{Position and line-of-sight stabilizing controller}
Without loss of generality we assume that the desired position of
the robot is the origin and the line-of-sight is $Z_b$. The control
objective is to stabilize the position of the robot to the origin  and
the line-of-sight (fixed to the body) $Z_b$ to coincide with the $Z_i$-axis of the
inertial frame. In other words, the objective is to stabilize the closed loop system to submanifold $E=\{(x,y,R,\omega)\in \R^2\times \SO{3}\times \R^3 :x=0$, $y=0$ and
$\omega= R^\top e_3\}$. We note that $\omega= R^\top e_3 \Rightarrow \dot{\omega}=0$.

Before we proceed to derive the control to meet the aforementioned
objective, consider the configuration  error function
$\psi:\mathbb{R}^2\rightarrow \mathbb{R}$
\beqn
\psi(x,y)=k_p(x^2+y^2), k_p>0 \; \text{is\;\;free}.
\eeqn
Using $\psi$, the position of the robot can be stabilized to the
origin of the $(X,Y)$ plane. The controller synthesis can proceed as
follows.

 The derivative of $\psi(x,y)$ with respect to time
along the trajectories of \eqref{kin} is given by,

\beq
\begin{array}{lcl}
\frac{d}{dt}\psi(x,y)&=&k_p(x\dot{x}+y\dot{y}) \\
&=&k_p(xr_2-yr_1)\cdot \omega
\end{array}\label{error_function_posi}
\eeq
   Equation
\eqref{error_function_posi} can be rewritten as
\beqn
\begin{array}{lcl}
\frac{d}{dt}\psi(x,y)&=&k_p(x\dot{x}+y\dot{y})\\
&=&k_pr(x(r_2\cdot\omega-r_2\cdot r_3)-y(r_1\cdot\omega-r_1\cdot r_3))\\
&=&k_pr(xr_2\cdot(\omega-r_3)-yr_1\cdot(\omega-r_3))\\
&=&k_pr(x r_2-y r_1)\cdot(\omega-r_3)
\end{array}
\eeqn
which implies that $\frac{d}{dt}\psi=d\psi e_{\omega}$, where
$e_{\omega}\deff (\omega-r_3)$ is the velocity error. Hence the
error function $\psi$ is compatible with $e_{\omega}$. If
$\omega_d=r_3$, then $\Omega_d=e_3$, where the subscript $d$ refers
to the desired values.

The right transport map ${\cal T}: \SO{3} \times T_{R_d}\SO{3}
\rightarrow T_{R}\SO{3}\times \SO{3}$ is defined as
\beqn {\cal T}(R,R_d)(\dot{R}_d)= \dot{R}_dR_d^\top R. \label{rtm}
\eeqn
 Here, $\widehat{\Omega}_d=\dot{R}_dR_d^\top$ and $R_d$ satisfies
 $R_de_3=e_3$. Next, we define the velocity error using the
 transport map ${\cal T}$.
 \beq
 \begin{array}{lcl}
 {\cal T}(\dot{R}_d)&=&\widehat{\Omega}_d R\\
 &=&RAd_{R^\top}\widehat{\Omega}_d.
\end{array}
 \eeq

 The following derivatives are useful in deriving the covariant
derivative of right transport map. For $\hat{v}\in \so{3}$,
\beqn\frac{d}{dt}\Ad_{R}\hat{v}
 &=& \frac{d}{dt}R \hat{v} R^\top \\
&=& R\left(R^\top \dot{R}\hat{v}-\hat{v}R^\top \dot{R}\right)R^\top \\
&=& R(\hat{\omega} \hat{v}-\hat{v} \hat{\omega})R^\top \\
&=& R[\hat{\omega} ,\;\hat{v}]R^\top \\
&=& \Ad_{R}[\hat{\omega},\; \hat{v}]
 \eeqn
and $\frac{d}{dt}\Ad_{R^\top R_d}\hat{\omega}_d$ can be
expressed as
 \beqn
\begin{array}{lcl}&=& \left(\frac{d}{dt}(R^\top R_d)\right)\hat{\omega}_d  (R^\top_dR)+(R^\top R_d)\hat{\omega}_d\left(\frac{d}{dt}(R^\top_dR)\right)\\ &&+\Ad_{R^\top R_d}\hat{\dot{\omega}}_d \\
%&=&
%\left((R^\top R_d)\hat{\omega}_d (R^\top R_d)^\top \right)(R^\top \dot{R})\\ &&-(R^\top \dot{R})\left((R^\top R_d)\hat{\omega}_d (R^\top R_d)^\top  \right)\\ &&+\Ad_{R^\top R_d}\hat{\dot{\omega}}_d \\
&=&
\left((\Ad_{R^\top R_d}\hat{\omega}_d \right)(R^\top \dot{R})-(R^\top \dot{R})\left(\Ad_{R^\top R_d}\hat{\omega}_d \right)\\ &&+\Ad_{R^\top R_d}\hat{\dot{\omega}}_d \\
&=& \left[\Ad_{R^\top
R_d}\hat{\omega}_d,\hat{\omega}\right]+\Ad_{R^\top
R_d}\hat{\dot{\omega}}_d.
\end{array}
\eeqn
Thus, the covariant derivative of the  right transport map\\
$\covdg_{\dot{R}} {\cal T}(\dot{R}_d)$  is
 \beq
&=&\covdg_{\dot{R}}RAd_{R^\top }\widehat{\Omega}_d \nonumber\\
&=&R\left(\frac{d}{dt}Ad_{R^\top}\widehat{\Omega}_d+\covdgi_{\hat{\omega} }Ad_{R^\top }\widehat{\Omega}_d\right)\nonumber\\
&=&R\left(\left[Ad_{R^\top}\widehat{\Omega}_d,\hat{\omega}\right]+Ad_{R^\top}\widehat{\dot{\Omega}}_d+\covdgi_{\hat{\omega}}Ad_{R^\top
}\widehat{\Omega}_d\right)\nonumber\\
&=&R\left(\left[Ad_{R^\top}\widehat{\Omega}_d,\hat{\omega}\right]+\covdgi_{\hat{\omega}}Ad_{R^\top
}\widehat{\Omega}_d\right)\nonumber\\
&=&R \widehat{f}_{ff} \label{ff} \eeq
The last step follows by noting that ${\Omega}_d=e_3$.

We next present the feedforward and proportional-derivative
controller in $\mathbb{R}^3$.  For $v,\omega \in \mathbb{R}^3$, the
following holds
\beqn \left(\covdgi_{\hat{v}}\hat{\omega}\right)^{\vee} &=&
\frac{1}{2}(v\times \omega)+\frac{1}{2}J^{-1}\left(v \times
J\omega-Jv \times \omega\right) \eeqn
and from \eqref{so3_R3_iso} it follows
 \beqn
\left(\hat{v}\hat{\omega}-\hat{\omega}\hat{v}\right)=\left[\hat{v},
\hat{\omega}\right]_{\so{3}} &=&
\widehat{\left[v,\omega\right]}_{\mathbb{R}^3}=\widehat{\left(v
\times \omega\right)}. \eeqn
Thus $f_{ff}$ in \eqref{ff}  and $f_{pd}$ can be written as

%
% In step two, when $(x,y)=(0,0)$ and spatial angular
%velocity $\omega_s=e_3$ the remaining manifold is $\mathbb{S}^1$ i.e
%rotation along z-axis i.e $e_3$ to make $R=I$. Now we can use
%\eqref{dynamics}  $f_{pd}= -I^{-1}\left(k_p \left(skew(R_d^\top
%R(t))\right)^V+k_v(\omega-r_3)\right)$. And we derive the
%feed-forward term $f_{ff}$ as follows, %
%
\beq
\begin{array}{lcl}
 f_{ff}&=&R^\top e_3\times\omega+
\frac{1}{2}\left(\omega\times R^\top e_3\right.\\
&&\left.+J^{-1}\left(\omega\times J R^\top e_3-J\omega \times R^\top
e_3\right)\right)\\
\end{array}
\label{ff_control} \eeq
\beq
\begin{array}{lcl}
f_{pd}&=&-J^{-1}(k_pd\psi+k_v e_{ \omega})\\
&=&-J^{-1}(k_p r R^\top(x e_2-y e_1)\\
&&+k_v (\omega-R^\top e_3)).
 \end{array}
\label{pd_control} \eeq
With $\tau=J(f_{ff}+f_{pd})$, the closed-loop dynamics \eqref{dyn},
\eqref{ff_control} and \eqref{pd_control} is
 \beq \dot{\omega} &=& -J^{-1} (\omega\times J \omega ) + f_{ff}+f_{pd}.\label{closed_loop_dyn} \eeq
%
%%%---------------------------Lyapunov stability proof-------------
\begin{prop} Consider a spherical robot satisfying assumption \ref{assum}. Then, the closed-loop
system \eqref{closed_loop_dyn} is  asymptotically stable with respect to $(x,y, R^\top e_3)$ uniformly in $\omega$.
\end{prop}
\begin{proof}
Let $e_R=R\hat{e}_w$, Consider the candidate Lyapunov function
 \beqn
V&=&\frac{1}{2}\mathbb{G}(R)(e_R,e_R)+\psi(x,y)\\
&=&\frac{1}{2}\mathbb{G}(R)(R\hat{e}_{\omega},R\hat{e}_{\omega})+\psi(x,y)\\
&=&\frac{1}{2}\mathbb{G}(I)(\hat{e}_{\omega}(I),\hat{e}_{\omega}(I))+\psi(x,y).
\eeqn
 The derivative of $V$ with respect to
time  along the trajectories of the closed-loop system
\eqref{closed_loop_dyn} is
\beqn
\dot{V}&=&\mathbb{G}(I)\left(\hat{e}_{\omega}(I),\covdg_{\hat{\omega}}\hat{e}_{\omega}(I)\right)+\dot{\psi}(x,y)\\
&=&\mathbb{G}(I)\left(\hat{e}_{\omega}(I),\covdg_{\hat{\omega}}(\hat{\omega}-Ad_{R^\top}\widehat{e}_3)(I)\right)+\dot{\psi}\\
&=&G(I)\left(\hat{e}_{\omega}(I),(\covdg_{\hat{\omega}}\hat{\omega}-\covdg_{\hat{\omega}}Ad_{R^\top}\widehat{e}_3)(I)\right)+\dot{\psi}\\
&=&J\left(e_{\omega},(\frac{d}{dt}\hat{\omega}+\covdgi_{\hat{\omega}}\hat{\omega}-\hat{f}_{ff})^{\vee}\right)+\dot{\psi}\\
%&=&J\left(\hat{e}_{\omega}(I),\hat{f}_{pd}\right)+\dot{\psi}(x(t),y(t))\\
&=&J\left(e_{\omega},f_{pd}\right)+\dot{\psi}\\
&=&J\left(e_{\omega},-J^{-1}(d\psi+k_v e_{\omega})\right)+\dot{\psi}\\
&=&I\left(e_{\omega},-d\psi-k_v e_{\omega}\right)+\dot{\psi}\\
&=&-k_ve_{\omega}^\top e_{\omega}-e_{\omega}^\top d\psi+\dot{\psi}\\
&=&-k_ve_{\omega}^\top e_{\omega}\le 0.
%&=&J(\hat{e}_{\omega}(I),-J^{-1}(\widehat{d\psi}+k_v\hat{e}_{\omega})(I))+\dot{\psi}(x(t),y(t))\\
%&=&\left<\hat{e}_{\omega}(I),-(\hat{d\psi}+k_v\hat{e}_{\omega})\right>_{Tr}+\dot{\psi}(x(t),y(t))\\
%&=&\left<\hat{e}_{\omega}(I),-\widehat{d\psi}\right>_{Tr}+\left<\hat{e}_{\omega}(I),-k_v\hat{e}_{\omega}\right>_{Tr}+\dot{\psi}(x(t),y(t))\\
%&=&\left<e_{\omega},-d\psi\right>_{\mathbb{R}3}-k_v\left<e_{\omega},e_{\omega}\right>_{\mathbb{R}^3}+\dot{\psi}(x(t),y(t))\\
%%&=&\left<e_{\omega},-d\psi\right>\\
%&=&-k_v\left<e_{\omega},e_{\omega}\right>_{\mathbb{R}^3}\le 0.
%
\eeqn
Let $L\deff \{(x,y,R,\omega)\in \R^2\times \SO{3}\times
\R^3:V(x,y,R,\omega)\le c, c>0\}$ is compact, connected and contains
$E$. Consider the residual set $S\deff\{(x,y,R,\omega)\in
L:\dot{V}=0\}$. Let $(x,y,R,\omega)\in S \implies
 \omega=R^\top e_3,\dot{\omega}=0$. Since  $r_1$ and $r_2$ are independent, from \eqref{error_function_posi} it follows that $k_p(xr_2-yr_1)=0$
if and only if  $x=0$ and $y=0$. Thus the largest
invariant set in $S$ is $E$. Thus, by LaSalle's invariance
principle, all trajectories originating in $L$ approach $E$
asymptotically.
\end{proof}
%%%---------------------------

%
Thus the controller stabilizes the robot to the origin of the
$(X,Y)$ plane at which the robot spins about its local vertical axis
($Z_b$-axis) at a constant angular velocity.
\section{SIMULATIONS}
The system parameters used for simulation is $r=0.4\
\mathrm{m},J=\mathrm{diag}(0.3,0.4,0.5)\ \mathrm{kg} \mathrm{m}^2$.
The control gains in \eqref{closed_loop_dyn} are chosen as
$k_p=5,k_v=1$.
The time-response of the closed-loop  with the initial condition
$x(0)=4\ \mathrm{m},y(0)=3\ \mathrm{m},
R(0)=\left[\begin{array}{ccc}
1 & 0 & 0 \\
 0 & \frac{1}{\sqrt{2}} & -\frac{1}{\sqrt{2}} \\
 0 & \frac{1}{\sqrt{2}} & \frac{1}{\sqrt{2}}
\end{array}  \right],
\omega (0) = (0,0,0)\ \mathrm{rad/s} $ is shown in Figure
\ref{fig:dynstab} and the $(x,y)$ trajectory is shown in Figure
\ref{fig:dynstab2}.

\begin{figure}
\centering
\includegraphics[scale=0.9]{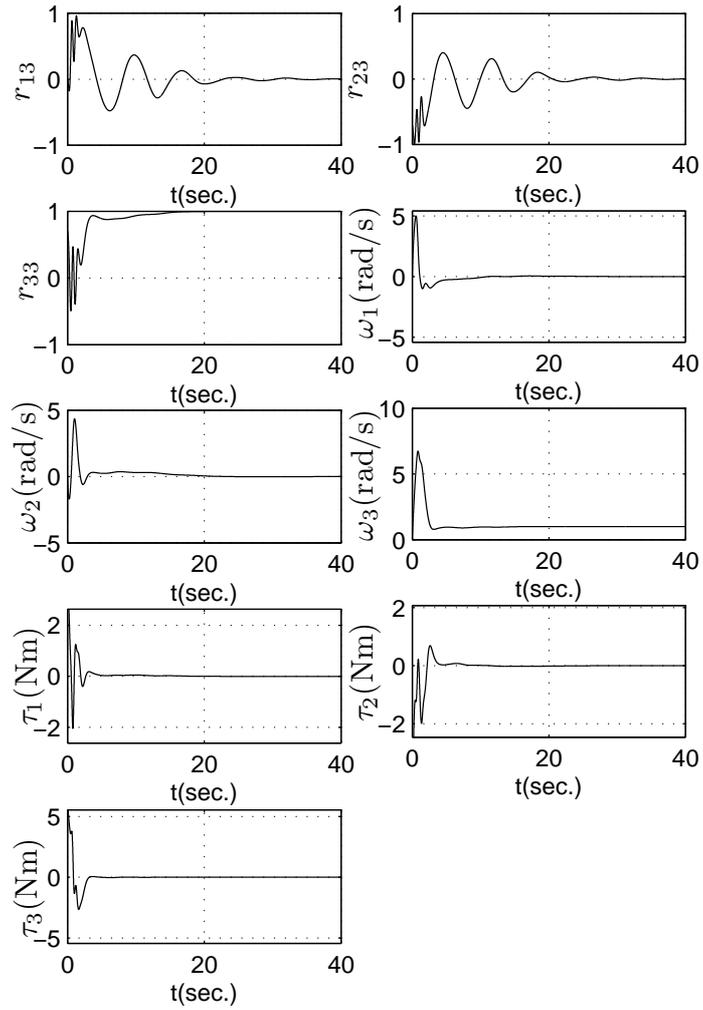}
      \caption{Time-response of attitude dynamics}
\label{fig:dynstab}
\end{figure}
\begin{figure}
\centering
\includegraphics[scale=0.7]{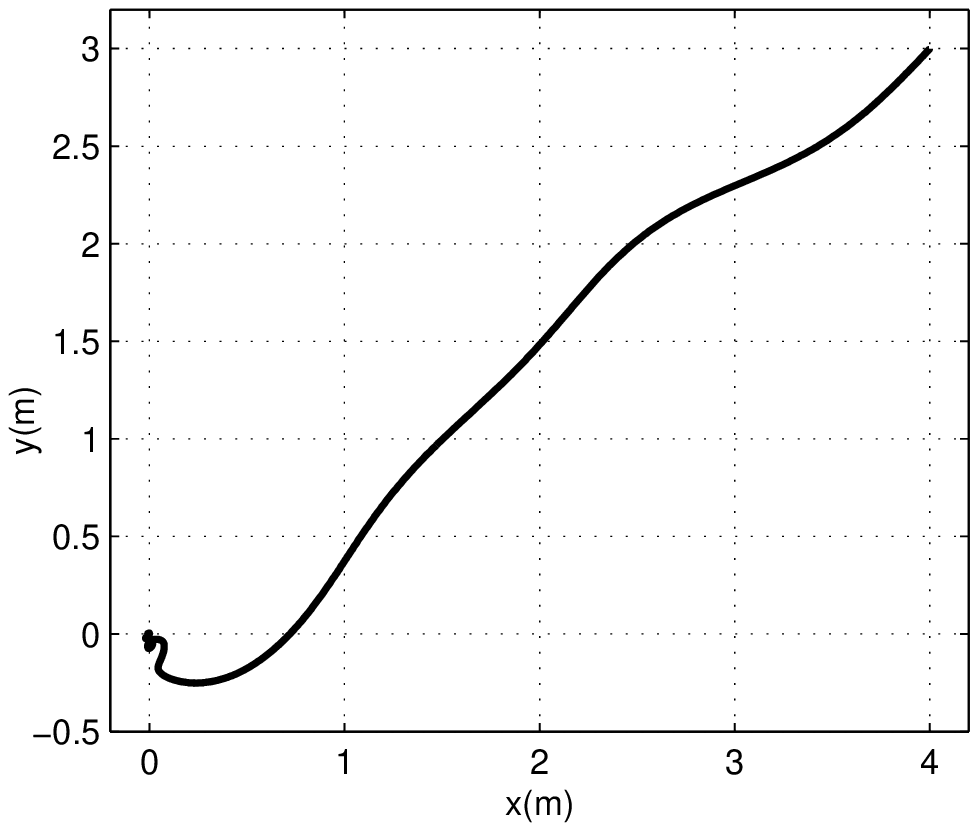}
      \caption{$(x,y)$ trajectory}
\label{fig:dynstab2}
\end{figure}

The simulation is repeated with   $R(0)=\left[\begin{array}{ccc}
1 & 0 & 0 \\
 0 & -1 & 0 \\
 0 & 0 & -1
\end{array}  \right] $ while all other initial condition remaining the same.  The time-response is shown in Figure
\ref{fig:dynstab3} and the $(x,y)$ trajectory is shown in Figure
\ref{fig:dynstab4}.

\begin{figure}
\centering
\includegraphics[scale=0.9]{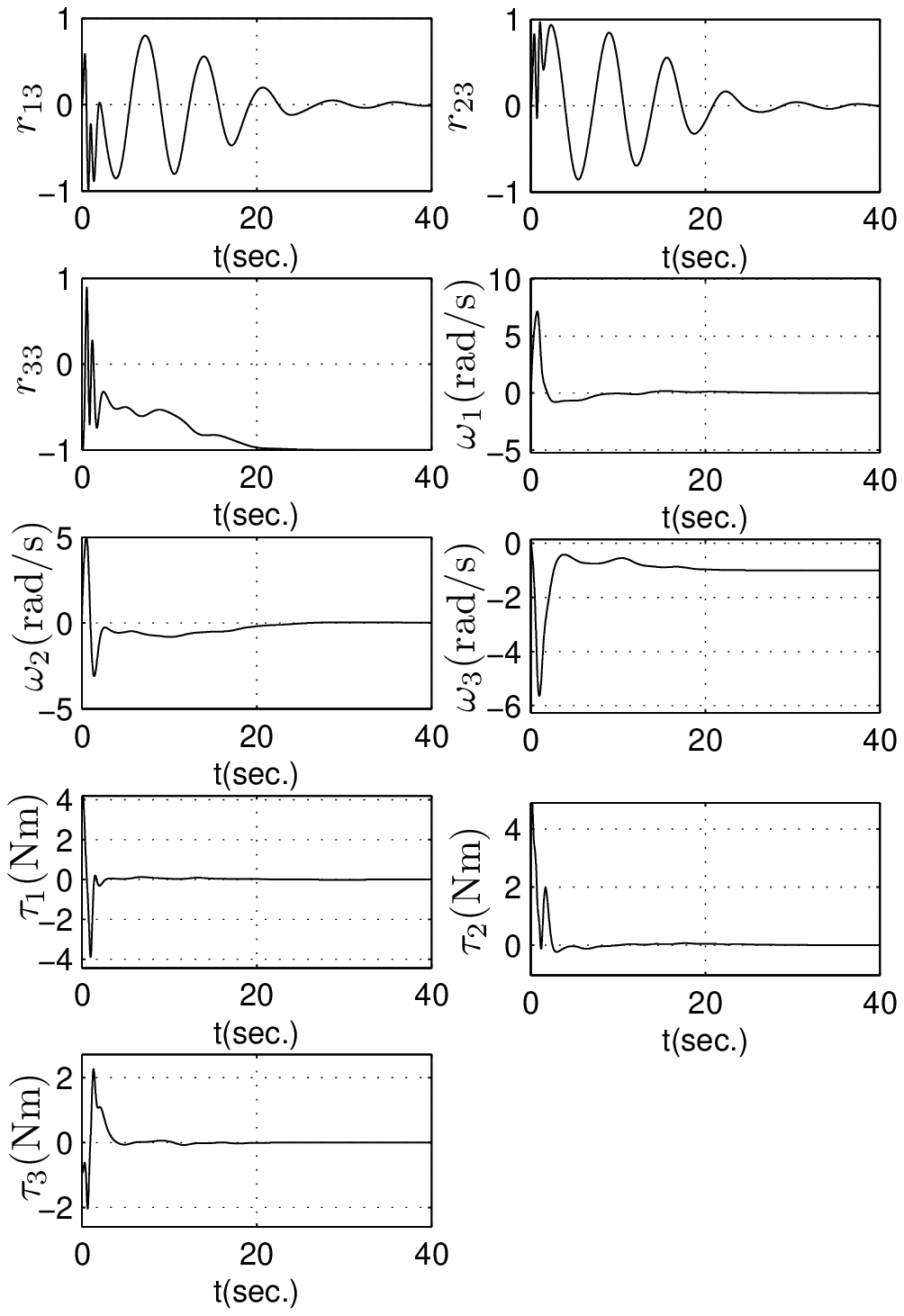}
      \caption{Time-response of attitude dynamics}
\label{fig:dynstab3}
\end{figure}
\begin{figure}
\centering
\includegraphics[scale=0.7]{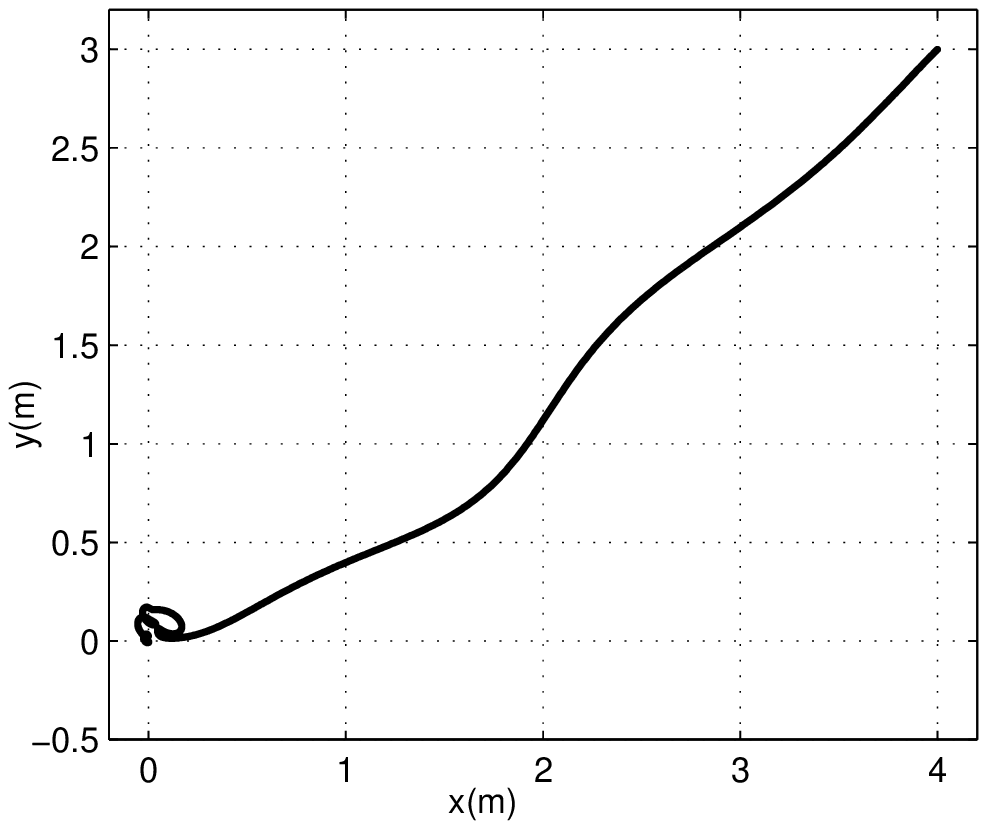}
      \caption{$(x,y)$ trajectory}
\label{fig:dynstab4}
\end{figure}

A consequence of the control law  is the regulation of $\Omega$ to
$e_3$, which implies a) $\omega=e_3,r_3=e_3$ as seen in Figure
\ref{fig:dynstab} or  b) $\omega=-e_3,r_3=-e_3$ as seen in Figure
\ref{fig:dynstab3}.

\section{Conclusions}
In this paper we have presented a smooth geometric controller to
asymptotically stabilize the system to a smooth submanifold. This results in
the robot reaching the origin of the plane while the robot spins
with constant angular velocity about its local spin-axis, which by
design is the body $Z_b$-axis coincident with the inertial
$Z_i$-axis. This control strategy can be used in line-of-sight
application for payload pointing, such as a camera mounted inside
the sphere.
%-------------------------------------------
%\addtolength{\textheight}{-6cm}
%
\bibliographystyle{authoryear}
\bibliography{ref}
\end{document}